\documentclass[conference]{IEEEtran}
\IEEEoverridecommandlockouts
\usepackage{cite}
\usepackage{amsmath,amssymb,amsfonts, amsthm}
\usepackage{graphicx}
\usepackage{subcaption}
\usepackage{textcomp}
\usepackage{xcolor}

\usepackage{comment}
\usepackage{mathtools}
\usepackage[english]{babel}
\newtheorem{thm}{Theorem}

\usepackage{algorithm}
\usepackage[noend]{algpseudocode}

\def\BibTeX{{\rm B\kern-.05em{\sc i\kern-.025em b}\kern-.08em
		T\kern-.1667em\lower.7ex\hbox{E}\kern-.125emX}}
\begin{document}
	
	\title{Low-Complexity Integer Divider Architecture for Homomorphic Encryption
		\thanks{This work is supported in part by Cisco Systems.}
	 }
	\author{\IEEEauthorblockN{Sajjad Akherati, Jiaxuan Cai, and Xinmiao Zhang}
	 	\IEEEauthorblockA{The Ohio State University, OH 43210, U.S.
	 }}
	
	\maketitle
	\begin{abstract}
		
		Homomorphic encryption (HE) allows computations to be directly carried out on ciphertexts and enables privacy-preserving cloud computing. The computations on the coefficients of the polynomials involved in HE are always followed by modular reduction, and the overall complexity of ciphertext multiplication can be reduced by utilizing the quotient. Our previous design considers the cases that the dividend is an integer multiple of the modulus and the modulus is in the format of $2^w-2^u\pm1$, where $u<w/2$. In this paper, the division is generalized for larger $u$ and dividend not an integer multiple of the modulus. An algorithm is proposed to compute the quotient and vigorous mathematical proofs are provided. Moreover, efficient hardware architecture is developed for implementing the proposed algorithm. Compared to alternative division approaches that utilize the inverse of the divisor, for $w=32$, the proposed design achieves at least 9\% shorter latency and 79\% area reduction for 75\% possible values of $u$.
	\end{abstract}
	
	
	\section{Introduction} 
	Homomorphic encryption (HE) allows computations to be carried out on ciphertexts without decryption. It is the key to preserving privacy in cloud computing. Popular HE schemes \cite{BGV, BFV, CKKS} involve computations over very long polynomials, whose coefficients are large integers, and coefficients multiplication and addition are followed by modular reduction. To reduce the computation complexity, the large coefficients are represented by residue number system \cite{RNSCKKS}, and moduli with a small number of nonzero bits, such as in the format of $q=2^w-2^u\pm 1$, are chosen. It was found in \cite{HE} that the overall complexity of ciphertext multiplication can be reduced by combining and reformulating the coefficients multiplication and relinearization, which is enabled by using the quotient of dividing coefficients product by $q$.
	
	The division can be implemented by using a look-up table \cite{lokuptab}. However, the size of the look-up table increases exponentially with the number of bits to divide in each clock cycle. Approximate division by very short integers, such as 8-bit, has been investigated for image processing in \cite{Canny, edgeDet}. The approximations in these schemes lead to a big difference in the quotient for large $q$. 
	In \cite{recdiv}, the quotient is derived by multiplying an approximation of $q^{-1}$. To improve the precision, $2w$ bits are used to represent the approximation when $q$ has $w$ bits. The dividend is a product of two $w$-bit coefficients and also has $2w$ bits. Hence, a $2w\times 2w$ wide multiplier is needed, and it leads to not only a long data path but also a large area. To address these issues, the quotient is calculated as $a\times \lambda+b$, and then the least significant bits are deleted in \cite{cra}. Here $a$ and $b$ are precomputed constants with at most $w$ bits. Although the width of the multiplicand is reduced, a wide multiplier is still needed for this design.

The design in \cite{HE} assumes that $q=2^w-2^u\pm1$. Utilizing the property that $q$ has a small number of nonzero bits, the quotient is calculated by addition and shift operations that have much shorter data path and smaller area requirements compared to those multiplying approximation of $q^{-1}$ as in \cite{recdiv,cra}. However, the design in \cite{HE} is limited to the case of $u<w/2$, and the dividend is an integer multiple of $q$. Given the product of two coefficients, its remainder of division by $q$ needs to be calculated and subtracted first before the division can be carried out. 

This paper proposes a generalized low-complexity integer division algorithm and implementation architecture. An iterative process is developed to compute the quotient in the case of $u\geq w/2$, where each iteration consists of simple addition and shift operations. The number of iterations needed is a small value depending on the ratio of $u/w$. Mathematical formulas and corresponding proofs are given for the number of iterations. Unlike the algorithm in \cite{HE}, the proposed design does not require the dividend to be an integer multiple of $q$ and hence does not need a separate remainder calculation. Instead, the quotient computed from the iterative process is adjusted to take into account the remainder. Efficient hardware implementation architectures are also developed for the proposed algorithm and synthesis has been carried out. For $w=32$, there are 31 different possible $u$. For 50\%, 16\%, and 9\% of these possible values of $u$, the proposed design achieves 55\%, 32\%, and 9\% shorter latency, respectively, and at least 79\% silicon area reduction compared to the divider in \cite{cra}. 
	
\section{Existing Integer Dividers}
Consider $\lambda/q$ where $\lambda$ and $q$ have $2w$ and $w$ bits, respectively. In \cite{recdiv}, the quotient of $\lambda/q$ is calculated by multiplying $\lambda$ with a pre-computed constant $J=\lfloor2^{3w-1}/q\rfloor+1$. The higher $w$ bits of this product is the quotient of $\lambda/q$. $J$ has $2w$ bits in order to ensure the correctness of the quotient. The design in \cite{cra} calculates the quotient as $\lfloor(a\lambda+b)/2^k\rfloor$, where $a$, $b$ and $k$ are constants pre-computed according to $q$. $a$ is an ($k-w$)-bit number, where $w<k\leq 2w$. Consequently, the width of the multiplier needed is reduced. However, the quotient computed from this algorithm can be different from the actual quotient by $\pm 1$. 
	
The design in \cite{HE} assumes $q=2^w-2^u\pm1$, $u< w/2$, and $\lambda$ is an integer multiple of $q$. Let $c=\lfloor\lambda/2^w\rfloor$. It first calculates $b^* = c - \lfloor\frac{-c(2^u\mp 1)}{2^w}\rfloor$. It was shown that $\lambda/q$ equals $b^*$ when $b^*$ and $\lambda$ are both even or odd. Otherwise, $\lambda/q=b^*+1$. As a result, the division was implemented by two adders and two shifters. 

	\section{Generalized Low-Complexity Integer Division }
This section first generalizes the previous division algorithm in \cite{HE} to the case of $q=2^w-2^u\pm1$ with $u\geq w/2$. Then the design is further extended to the case that $\lambda$ is not an integer multiple of $q$. 

\subsection{Extension for $u\geq w/2$ with $\lambda$ as an integer multiple of $q$}
Denote the quotient of $\lambda/q$ by $b$ ($b\in Z^+$). Replace $q$ by $2^w-2^u\pm1$ in $\lambda = bq$. Then $c=\lfloor\lambda/2^w\rfloor$ can be rewritten as $c=\lfloor\frac{bq}{2^w}\rfloor = b + \lfloor(-1)\frac{b(2^u\mp1)}{2^w}\rfloor$.
Define 
\begin{equation}\label{fx}
f(x) \triangleq x+v(x),
\end{equation}
where 
\begin{equation}\label{vx}
v(x) \triangleq \lfloor(-1)\frac{x(2^u\mp1)}{2^w}\rfloor.
\end{equation}
$b$ is the solution of $f(x)=c$. This solution can be found iteratively using Algorithm \ref{alg}. In this algorithm, the least significant bit (LSB) of a number $y$ is denoted by $LSB(y)$. 
	\begin{algorithm}[H]
		\caption{Algorithm for calculating $b=\lambda/q$ ($\lambda$ is an integer multiple of $q$)} \label{alg}
		\hspace*{\algorithmicindent} \textbf{Input:} $\lambda$, $q=2^w-2^u\pm1$, $w$, $u$
		\begin{algorithmic}[1]
			\State $c\gets \lfloor\frac{\lambda}{2^w}\rfloor$; $i\gets 0$; $b_0 \gets c$;
			\While{$f(b_i)\neq c$}
			\State $b_{i+1}\gets b_i + (c-f(b_i))$; $i\gets i+1$; $b^*\gets b_i$;
			\EndWhile
			\State $b\gets b^* + (LSB(\lambda)\ \text{XOR}\ LSB(b^*))$;
			\State \Return $b$;
		\end{algorithmic}
	\end{algorithm}
	
In the following, Theorem 1 proves that the loop in Algorithm 1 will terminate with a finite number of iterations. Theorem 2 connects the $b^*$ computed from the loop with the actual $b$ value. The number of iterations needed in Algorithm 1 is given through Theorem 3 and 4.

\begin{thm}\label{t1} Algorithm \ref{alg} will terminate with a finite number of iterations.
	\end{thm}
	\begin{proof} 
By induction, it is shown in the following that $f(b_{i-1})<f(b_{i})\leq c$. 	Using the properties of the floor function, it can be derived that $f(x)+f(y) \leq f(x+y) \leq f(x)+f(y)+1$. Besides, from \eqref{vx}, it can be easily seen that $0< f(x) < x$ for $x>0$. Hence, for $i=1$:
		\begin{align*}
			f(b_1)\!=\! f(b_0\!+\!c\!-\!f(b_0)) 
			&\!\geq \!f(b_0)\!+\! f(c\!-\!f(b_0))\\&\!=\!f(b_0)\!+\!f(c\!-\!f(c))\!>\!f(b_0)
		\end{align*}
Similarly, \small
		\begin{align*}
			f(b_1)\!\leq \!f(b_0)\!+\!f(c\!-\!f(b_0)) \!+\! 1\!<\! f(b_0)\!+\!(c\!-\!f(b_0)) \!+\! 1 \!=\!c\!+\!1.
		\end{align*}\normalsize
Since $f(b_1)$ is integer, $f(b_1)\leq c$. Assume that $f(b_{i-2})<f(b_{i-1})\leq c$, similar derivations as in the above two formulas show that $f(b_{i-1})<f(b_{i})\leq c$. Since $f(b_i)$ strictly increases with iteration $i$ and it is always an integer not exceeding $c$, $f(b_i)$ will equal to $c$ in an iteration and the loop in Lines 2 and 3 of Algorithm 1 terminates.
\end{proof}
	
	\begin{thm}\label{t2}
The $b_i$ in every iteration of Algorithm \ref{alg} is at most $b$.
	\end{thm}
	\begin{proof} This is proved by induction. Since $c=b+v(b)$ and $v(b)\leq 0$, it is clear that $b_0=c\leq b$. Now suppose that $b_i\leq b$. From Algorithm \ref{alg}, if $b_i=b$, then the loop terminates and there will be no other $b_i'$ with $i'>i$. If $b_i<b$, then $v(b)\leq v(b_i)$ from \eqref{vx}. Accordingly,
		\begin{align*}
			b_{i+1} &= b_{i} + c -f(b_{i}) = b_{i} +c - (b_{i}+v(b_{i}))\\
			&= c-v(b_{i})= b+v(b)-v(b_{i})\leq b.
		\end{align*}
\end{proof}
Theorem \ref{t2} shows that the $b^*$ calculated by the loop in Algorithm 1 does not exceed $b$. Next, it will be shown that $b^*$ equals either $b-1$ or $b$. From \eqref{fx}, $f(b)=b+v(b)$ and $f(b-1)=b-1+v(b-1)$. Since $v(b-1)=v(b)$ or $v(b)+1$, $f(b-1) = f(b)$ or $f(b)-1$. If $f(b-1)=f(b)=c$ then the $b^*$ from the loop in Algorithm \ref{alg} may be either $b-1$ or $b$. Similarly, it can be shown that $f(b-2)=f(b)-1 \text{ or }f(b)-2$. Hence, $f(b-2)\neq c$ and $b^*$ can not be $b-2$ or a smaller value. Since $q$ is an odd number, $b$ is even or odd when $\lambda$ is even or odd, respectively. Hence, whether $b$ equals $b^*$ or $b^*+1$ can be decided by using the LSB of $\lambda$ and $b^*$ as listed in Line 4 of Algorithm \ref{alg}. 

The initial value of $b_0$ should be less than $b$, which is unknown at the beginning of Algorithm 1. On the other hand, initializing $b_0$ to the smallest positive integer, 1, leads to more iterations in Algorithm 1. Since $c\leq b$, $c$ is chosen as the initial value in our algorithm. With this initial value, $c-f(b_i)$ is the largest value that can be used to update $f(b_i)$ as proposed in Algorithm 1. Consider the case that $b=2$ and $u<w-2$. Then $c=1$ and  $f(b_0) = f(c) = f(1) = 0$. If $c-f(b_0)+j$ ($j\geq 1$)  is used to update $b_i$. Then $b_1 = b_0 + (c-f(b_0)+j) \geq 3$ and $f(3)=2>c$. Hence, Algorithm \ref{alg} will not terminate and $c-f(b_i)+j$ with $j\geq 1$ can not be used to update $f(b_i)$.

	\begin{thm}\label{t3}
Suppose that the number of iterations to calculate $b=\frac{\lambda}{q}$ ($\lambda<2^{2w}-1$) and $b_{\text{MAX}}\triangleq \lfloor\frac{2^{2w}-1}{q}\rfloor$ from Algorithm \ref{alg} are $N$ and $N_{\text{MAX}}$, respectively. Then $N\leq N_{\text{MAX}}$.
	\end{thm}  
	\begin{proof} Let $c_\text{MAX}=\lfloor (2^{2w}-1)/2^w\rfloor$ and $b_{{\text{MAX}}_i}$ be the intermediate value at iteration $i$ of Algorithm 1 for calculating $b_{\text{MAX}}$. It is shown below that $b-b_i \leq b_{\text{MAX}}-b_{{\text{MAX}}_i}$ by induction. It can be derived that
 \small
		\begin{align*}
			b-b_0 &= b - \lfloor\frac{bq}{2^w}\rfloor < b-\frac{bq}{2^w}+1=b(1-\frac{q}{2^w})+1\\&\leq b_\text{MAX}(1-\frac{q}{2^w})+1 = b_\text{MAX}-\frac{b_\text{MAX}q}{2^w}+1 \\&
< b_\text{MAX}-\lfloor\frac{b_\text{MAX}q}{2^w}\rfloor+1 = b_\text{MAX}-b_{{\text{MAX}}_0}+1.
		\end{align*} \normalsize
$b$ and $b_0$ are integers, and hence $b-b_0\leq b_\text{MAX}-b_{{\text{MAX}}_0}$. 

Suppose that for iteration $i < N$, $b-b_i \leq b_\text{MAX}-b_{{\text{MAX}}_i}$. For $i+1<N$
		\small
		\begin{align*}
			b-b_{i+1} &= b-(b_i+c-f(b_i))= b-c+v(b_i)\\&< b-c-\frac{b_i(2^u\mp1)}{2^w}< b-\frac{bq}{2^w}-\frac{b_i(2^u\mp1)}{2^w}+1 \\&= (b-b_i)\frac{2^u\mp1}{2^w}+1\leq (b_\text{MAX}-b_{\text{MAX}_i})\frac{2^u\mp1}{2^w}+1 \\&= b_{\text{MAX}}-\frac{b_{\text{MAX}}q}{2^w}-\frac{b_{\text{MAX}_i}(2^u\mp1)}{2^w}+1\\
   &= b_{\text{MAX}}-\lfloor\frac{b_{\text{MAX}}q}{2^w}\rfloor+\lfloor(-1)\frac{b_{\text{MAX}_i}(2^u\mp1)}{2^w}\rfloor\\&\quad+1 + \frac{-[b_{\text{MAX}}q]_{2^w}+[-b_{\text{MAX}_i}(2^u\mp1)]_{2^w}}{2^w},
		\end{align*}
		\normalsize
where $[\cdot]_{2^w}$ denotes the remainder of dividing by $2^w$. It can be shown that $0\!<\!1\! + \!\frac{-[b_{\text{MAX}}q]_{2^w}+[-b_{\text{MAX}_i}(2^u\mp1)]_{2^w}}{2^w}\! <\! 1$. Since $b-b_{i+1}$ is an integer,
\small
\begin{align*}
    b-b_{i+1} &\leq  b_\text{MAX}-c_\text{MAX}+v(b_{\text{MAX}_i}) = b_\text{MAX}-b_{\text{MAX}_{i+1}}.
\end{align*}
\normalsize

Algorithm 1 terminates at iteration $N$ when $\lambda$ is the input. From previous analysis, $b_N$ equals either $b$ or $b-1$. From the above proof, $b_{\text{MAX}}-b_{\text{MAX}_N} \geq b-b_N$. Hence $b_{\text{MAX}}-b_{\text{MAX}_N}\geq 0$. On the other hand, $b_{\text{MAX}}-b_{\text{MAX}_{N-1}}\neq 0$. Otherwise, the algorithm with $\lambda$ as the input already terminates at iteration $N-1$ and this contradicts the assumption. Accordingly, when the input of Algorithm 1 is $2^{2w-1}$, it may need more than $N$ iterations to compute $b_\text{MAX}$ and hence $N_{\text{MAX}}\geq N$.  

	\end{proof} 
	\begin{thm}\label{t4}
	For $q = 2^w-2^u\pm1$, let $t$ be the integer such that $tu>(t-1)w$ and $(t+1)u\leq tw$. Then $N_{\text{MAX}}=t$. 
	\end{thm}
	\begin{proof}
Clearly $c_\text{MAX}=2^w-1$. From Algorithm \ref{alg}
\small
		\begin{align*}
			b_{\text{MAX}_1} &= b_{\text{MAX}_0} + c_\text{MAX}-f(b_{\text{MAX}_0}) = c_\text{MAX}-v(c_\text{MAX})\\
			&= 2^w-1-\lfloor(-1)\frac{(2^w-1)(2^u-1)}{2^w}\rfloor= 2^w+2^u-2.
		\end{align*} 
  \normalsize
By substituting the above formula into Line 3 of Algorithm 1 for $i$ iterations, it can be derived that $b_{\text{MAX}_i} = 2^w+2^u+2^{2u-w}+2^{3u-2w}+\cdots + 2^{iu-(i-1)w}-\{1\text{ or }2\}$. It can be derived that $f(b_{\text{MAX}_i})\neq c_{\text{MAX}}$ for $i<t$, where $t$ is the integer such that $tu>(t-1)w$ and $(t+1)u\leq tw$. On the other hand, $f(b_{\text{MAX}_t})= c_{\text{MAX}}$. This means that Algorithm 1 needs $t$ iterations to terminate when its input is $2^{2w}-1$.
	\end{proof}

From Theorem 3 and 4, Algorithm 1 needs at most $t$ iterations, where $t$ is the integer such that $tu>(t-1)w$ and $(t+1)u\leq tw$. For $u\leq 3/4w$, Algorithm 1 terminates in at most 3 iterations. Algorithm 1 still applies when $u<w/2$, in which case $t=1$ and Algorithm 1 reduces to the same algorithm as in \cite{HE}.

\subsection{Extension for $\lambda$ not an integer multiple of $q$} 
Assume that $\lambda = bq+r$, where $0\leq r<q$. If $r=0$, it reduces to the case covered in Subsection A, and the loop in Algorithm \ref{alg} returns a $b^*$ that equals $b$ or $b-1$ from the previous analysis. Consider the case of $0<r<q$, $\lambda<\lambda'=(b+1)q$. Feeding $\lambda'$ as the input of Algorithm \ref{alg}, the loop would return a $b^*$ that is either $b+1$ or $b$. Therefore, $b$ may have three possible values: $b^*$, $b^*+1$ or $b^*-1$. 

Since $\lambda=bq+r$, $\lambda-b^*q$ should equal $r+q$, $r$, and $r-q$ when $b^*$ is $b-1$, $b$, and $b+1$, respectively. After $b^*$ is calculated, $\lambda-b^*q$ can be computed. Since $r-q<0<r<q<r+q$, $b$ is set to $b^*-1$, $b^*$, and $b^*+1$ when $\lambda-b^*q$ is less than 0, between 0 and $q$, and larger than $q$, respectively. As a result, the quotient of dividing a $\lambda$ that is not necessarily an integer multiple of $q$ can be calculated using Algorithm \ref{algg}. 

	\begin{algorithm}[H]
		\caption{Algorithm for calculating $\lfloor b=\lambda/q\rfloor$} \label{algg}
		\hspace*{\algorithmicindent} \textbf{Input:} $\lambda$, $q=2^w-2^u\pm1$, $w$, $u$
		\begin{algorithmic}[1]
			\State $c\gets \lfloor\frac{\lambda}{2^w}\rfloor$; $i\gets 0$; $b_0 \gets c$;
			\While{$f(b_i)\neq c$}
			\State $b_{i+1}\gets b_i + (c-f(b_i))$; $i\gets i+1$; $b^*\gets b_i$;
			\EndWhile
			\State $r^*\gets \lambda-b^*q$;       $b\gets b^*$;
               \If{$r^* < 0$} $b\gets b-1$;
            \ElsIf{$r^*\geq q$} $b\gets b+1$;
               \EndIf
			\State \Return $b$;
		\end{algorithmic}
	\end{algorithm}

\section{Divider hardware implementation architectures and comparisons}

This section first proposes efficient hardware architectures to implement the proposed divider. Then comparisons with other dividers are provided.

When $u<w/2$, Algorithm 2 has one iteration in the loop and it reduces to the algorithm proposed in \cite{HE}. The hardware architecture for implementing one iteration of Line 3 of Algorithm 2 is available in \cite{HE}. $b_{i+1} = b_i + c - f(b_i)$ can be simplified as $c+\lfloor\frac{b_i(2^u\mp1)}{2^w}\rfloor+D$. From \cite{HE}, $D=0$ when $b_i(2^u\mp 1)/2^w$ is an integer and $D=1$ otherwise. However, since $b_i$ is a $w$-bit number, it does not have $w$ factors of 2. $2^u\mp1$ does not have any factor of 2. Therefore, $b_i(2^u\mp 1)/2^w$ can not be an integer, and the architecture from \cite{HE} can be simplified to the units in the dashed block in Fig. \ref{division}. $t$ copies of these units are needed to implement $t$ iterations of the loop in Algorithm 2 in a pipelined manner. The shifters align the bits in the inputs to take care of the $2^u$ multiplication and $2^w$ division. The control signal $s$ is `1' and `0' when $q=2^w-2^u+1$ and $q=2^w-2^u-1$, respectively. $cin_i$ is set to `1' when the lower $u$ bits of $b_i$ are all `0' and $s=1$ in order to eliminate the addition on unnecessary bits. More explanations on these signals can be found in \cite{HE}.

	\begin{figure}[t]
		\centering
		\includegraphics[scale=0.85]{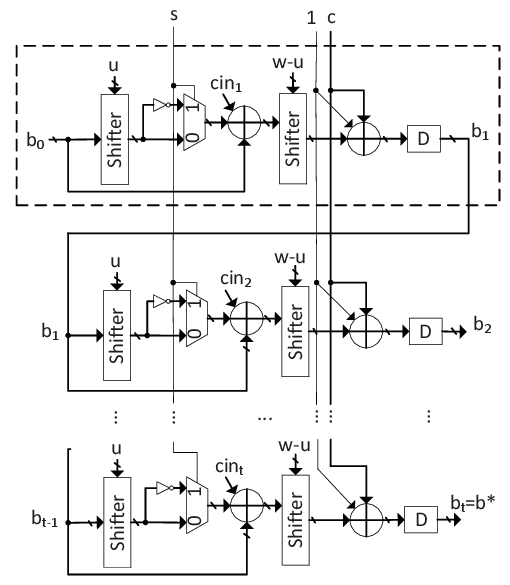}
		\caption{Architecture for calculating $b^*$ in Algorithm \ref{algg}.} \label{division}
	\end{figure}
 
From previous analysis, $\lambda-b^*q \in \{r, r\pm q\}$. Hence, the exact value of $\lambda-b^*q$ does not need to be calculated. It is sufficient to know if $\lambda-b^*q$ is negative, positive and less than $q$, or positive and larger than or equal to $q$. For $q=2^w-2^u\pm1$, the numbers that need to be added to calculate $\lambda-b^*q$ are shown in Fig. \ref{bars}(a). Since $ r< r+q<2^{w+1}$, if $\lambda-b^*q$ is positive, all its bits with weight at least $2^{w+1}$ are `0'. On the other hand, $|r-q|<2^w$. Hence, all the bits with weight at least $2^{w+1}$ are `1' in the 2's complement representation of $r-q$. Therefore, it is sufficient to tell that $\lambda-b^*q$ is negative if the $(w+2)$-th bit in $\lambda-b^*q$ is `1'. If this bit is `0', whether $\lambda-b^*q=r$ or $r+q$ can be decided by comparing the lower $w+1$ bits of $\lambda-b^*q$ with $q$. As a result, only the $w+2$ least significant bits of the numbers in Fig. \ref{bars} need to be added.

The architecture in Fig. \ref{bars}(b) computes $b$ from $b^*$. First `00' is padded to the left of the most significant bit (MSB) of $b^*$ to extend it to $(w+2)$-bit. The multiplexer on the top passes either 2's complement of $b^*$ or $b^*$ depending on whether the last number to add in Fig. \ref{bars}(a) is $-b^*$ or $+b^*$. The shifter in Fig. \ref{bars}(b) shifts the lower $w-u+2$ bits of the input to the left and pads $u$ `0's to the right. It aligns the bits from the $b^*$ number in the middle of Fig. \ref{bars}(a) for addition. Only the two LSBs of the $-b^*$ number in Fig. \ref{bars}(a) need to be added. They are padded with $w$ `0's to the right for the addition. The four numbers from Fig. \ref{bars}(a) are added up by the carry-save adder in the middle of Fig. \ref{bars}(b). The lowest $w+1$ bits of the adder output is compared with $q$. The value for $b$ is chosen based on the MSB of the adder output and the comparison result.
	
	\begin{figure}
		\centering
		\includegraphics[scale=0.95]{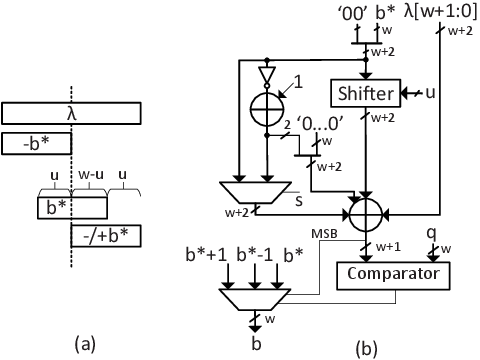}
		\caption{(a) The numbers need to be added for calculating $\lambda-b^*q$; (b) The hardware architecture to derive $b$ from $b^*$ using partial calculation of $\lambda-b^*q$.} \label{bars}
	\end{figure}
	
To further evaluate the complexity of our proposed design, it is synthesized using the Global Foundries 22FDX process for an example case of $w=32$ and $2w/3 \leq u < 3w/4$. In this case, $t=3$ copies of the units in Fig. \ref{division} are needed. Different timing constraints were tried in the synthesis, and the tightest timing constraint that does not lead to a substantial area increase is reported in Table \ref{tab1} in order to compare the minimum achievable clock period of different designs. In the proposed design, $t=3$ copies of units compute $b^*$ in three pipelining stages. After that, another clock cycle is needed to calculate $b$. 

\begin{table}[!h]
		\begin{center}\caption{Synthesis results of dividers with $w=32$ using Global Foundries 22FDX process}\label{tab1}
			\begin{tabular}{@{}c@{}||c|c|c@{}} 
				\hline
                & Timing  & Area  & Latency \\  
& constraint ($ps$)& ($\mu m^2$) & (\# of clks)\\\hline
                Design in \cite{recdiv} & $2600$ & $11168$ & $1$ \\ \hline
                Design in \cite{cra} & $1460$ & $6140$ & $1$\\ \hline 
                Proposed design ($t=3$) & $330$ & $1320$ & $4$\\\hline
			\end{tabular}
			\vspace{-1em}
		\end{center}
	\end{table} 

For comparison, the designs in \cite{recdiv} and \cite{cra} are synthesized and the results are also listed in Table \ref{tab1}. They have $2w\times 2w$ and $2w\times w$ multipliers, respectively, in their critical paths. On the other hand, the critical path of our proposed design only consists of a shifter, a $w$-bit carry-save adder, a comparator, and a few multiplexers as shown in Fig. \ref{bars}. As a result, our design achieves a much shorter clock period. The achievable improvement further increases for larger $w$. Although our design requires $t+1$ clock cycles to compute the quotient, the latency is still lower than that of the previous design due to the shorter clock period. For $w=32$, $t=1$ for $1\leq u<16$, $t=2$ for $16\leq u < 22$, and $t=3$ for $22\leq u <24$. Hence 75\% of possible $u$ leads to $t\leq 3$. Our proposed design achieves $1-330\times4/1460=9\%$, $1-330\times3/1460=32\%$ and $1-330\times2/1460=55\%$ latency reductions for $t=3$, 2, and 1, respectively, compared to the divider in \cite{cra}.

Since the designs in \cite{recdiv} and \cite{cra} consist of wide multipliers, their area requirement is much larger than that of the proposed design as shown in Table \ref{tab1}. The proposed architecture achieves 1-1320/6140=79\% area reduction compared to the design in \cite{cra}. For a $u$ corresponding to smaller $t$, fewer copies of the units in Fig. \ref{division} are utilized, and the achievable area reduction would be more significant.  
 
\section{Conclusions}
This paper proposed a low-complexity integer divider for calculating the quotient when the divisor has a small number of nonzero bits. It generalized the previous design to handle more possible divisors and the case that the dividend is not an integer multiple of the divisor. In addition, by analyzing the possible values of the intermediate results, simplifications on the hardware implementation architectures are developed. Compared to prior designs that are based on multiplying the inverse of the divisor, the proposed design reduces the area requirement to a fraction and also has much shorter latency. Future research will extend the proposed algorithm to the case that the divisor has more than three nonzero bits.

\end{document}